\documentclass[%
 aip,
 jmp,%
 amsmath,amssymb,
 reprint,%
 onecolumn,
author-numerical,%
]{revtex4-1}

\usepackage{graphicx}
\usepackage{dcolumn}
\usepackage{bm}

\usepackage{subcaption}
\usepackage{cancel}
\usepackage{amsfonts}
\usepackage{amsmath,amsthm}
\usepackage{bm}
\usepackage{amssymb}
\usepackage{mathrsfs}
\usepackage{amsmath, amsthm, amssymb,amsfonts,mathbbol,amstext}
\usepackage{mathtools}
\usepackage{mathrsfs}
\usepackage{dsfont}
\usepackage{graphicx}

\newtheorem{dfn}{\bf Definition}[section]

\newtheorem{thm}{\bf Theorem}[section]

\newtheorem{lem}{\bf Lemma}[section]

\newtheorem{prop}{\bf Proposition}[section]
\def\be{\begin{equation}}
\def\ee{\end{equation}}

\begin{document}


\title{Resource Theory of Contextuality}

\author{Barbara Amaral}
 \affiliation{Departamento de F\'isica e Matem\'atica, CAP - Universidade Federal de S\~ao Jo\~ao del-Rei, 36.420-000, Ouro Branco, MG, Brazil}
 
\date{\today}

\begin{abstract}
In addition to the important role of contextuality in foundations of quantum theory, this intrinsically quantum property has been identified as a potential resource for quantum advantage in different tasks. It is thus of fundamental importance  to study contextuality from the point of view of resource
theories, which provide a powerful framework for the formal treatment of a
property as an operational resource. In this contribution we review  recent developments towards a  resource theory of contextuality  and connections with operational applications of this property.
\end{abstract}

\maketitle

\section{Introduction}

Quantum theory provides a set of rules to predict probabilities of different outcomes in
different  settings. While it predicts probabilities which match with extreme
accuracy the data from actually performed experiments, quantum theory has some peculiar properties
which deviate from how we normally think about systems which have a probabilistic
description. One of these ``strange'' characteristics is the phenomenon of quantum contextuality \cite{KS67,AT18}, which implies that we cannot think about a measurement on a quantum system
as revealing a property which is independent of the set of measurements we choose to
make.  

Contextuality is one of the most striking features of quantum theory and, in addition to its role in the search for a deeper understanding of quantum theory itself, schemes that exploit this intrinsically quantum property
may be connected to the advantage of quantum systems over their classical counterparts.   Research in this direction has received
  a lot of attention lately and several recent  results highlight the power of quantum 
 relative to classical implementations.

Quantum contextuality is a necessary resource for
universal computing in models based on \emph{magic state distillation}
\cite{HWVE14}, in measurement  based quantum computation
 \cite{Raussendorf13,DGBR15} and also in computational models of qubits \cite{BDBOR17}; the presence of contextuality in a given system
 lower bounds the classical memory needed to simulate the experiment and in some situations reproducing the results of sequential measurements on a quantum  system exhibiting contextuality requires more memory than the information-carrying capacity of the system itself; contextuality can be used to
 certify the generation of genuinely random numbers \cite{PAMGMMOHLMM10, UZZWYDDK13}, a major problem in
 various areas, especially in cryptography;
  contextuality offers advantages in the problems of discrimination of states \cite{SSW18,AHLLBL19},
 one-way communication protocols \cite{SHP17}, and self-testing \cite{BRVWCK18}.

The identification of contextuality as a resource for several tasks motivated
 considerable  interest in \emph{resource theories}  for contextuality. Resource theories provide powerful structures for the formal treatment of a property
 as an operational resource, suitable for characterization, quantification, and manipulation \cite{CFS16}.
A resource theory consists essentially of three ingredients: a set of \emph{objects}, which represent the physical entities
which can contain the resource, and a subset of objects called \emph{free objects}, which are the objects that do not contain the resource; a special class of transformations, called \emph{free operations}, which fulfill the essential requirement
of mapping every free object of theory into a free object; finally, a \emph{quantifier}, which maps each object to a real number
that represents quantitatively how much resource this object contain, and which is monotonous under the action of
free operations.

 Recently, major steps have been taken towards the development of a unified resource theory for  contextuality, with the definition of a broad class of free operations with physical interpretation and explicit parametrization and of contextuality quantifiers that can be computed efficiently. In this contribution we review such developments and   connections with operational aspects of contextuality.

 The paper is organized as follows:  in section \ref{sec:resource} we describe the general framework of resource theories; in section \ref{sec:comp} we present the set of objects and free objects in the resource theory of contextuality considered here; in section \ref{sec:free} we study the set of non-contextual wirings, the most general class of free operations for contextuality for which an explicit parametrization is known; in section \ref{sec:quant} we list some of the most important contextuality quantifiers and their properties; in section 
 \ref{sec:app} we review some of the results connecting contextuality and quantum advantages in different tasks; we finish with a discussion and future perspectives in setion \ref{sec:conclusion}.

\section{Mathematical structure of a resource theory}
\label{sec:resource}

We start with a set of objects that  represent the physical entities that may possess the resource and a set of operations that transform one object into another. We assume that we can combine objects and compose operations in the intuitive way\footnote{We give more details below for the case of contextuality and the reader can find the general case in ref. \cite{CFS16}.} and that there is a trivial object $I$ such that $I$ combined with any object $A$ is equivalent to $A$.

The resource theory is defined when we choose a set of \emph{free objects} and \emph{free operations}, that intuitively represent the set of objects and operations that can be constructed or implemented at no cost.
 If the set of free operations is fixed, we say that an object $A$ is \emph{free} if there is a free operation $\mathrm{F}$ such that 
$ \mathrm{F}\left(I\right)=A$.
 If the set of free objects is fixed, we say that an operation $\mathrm{F}$ is \emph{free} if  
 $\mathrm{F}\left(A\right)$ is a free object for every free object $A$.

In this contribution the set of free objects will be defined using a
mathematical characterization of non-contextuality and the set of free operations
would, in principle, be the largest set of linear operations preserving this set. However, we still lack an explicit parametrization of this set, and we consider only a proper subset of free operations, namely, the \emph{non-contextual wirings} defined in ref. \cite{ACTA18}.

The set of free operations  defines a partial order in the set of objects.
 We say that $B\preceq A$ if there is a free operation $\mathrm{F}$ such
that $\mathrm{F}(A) = B$. A resource quantifier is a real function $Q$ defined in the set of objects that preserves the preorder, that is
\be B\preceq A \Rightarrow Q(B)\leq Q(A).\ee
Although resource quantifiers are useful in many situations, we remark that the preorder structure of a resource theory is more
fundamental than a single quantifier, unless the preorder is a total order \cite{CFS16}.

A simple example is the resource theory of beer, where the objects are glasses of beer. The trivial object is of course an empty glass and there is a unique free operation that corresponds to drinking beer. Other costly operations are filling the glass with more beer (costs you some money) or
producing beer (costs you some money, time, space, ...). In this case the preorder is a total order, because one can always compare the amount of beer in two glasses: we say that $G_1 \prec G_2$ if glass $G_1$ has less beer than glass $G_2$. The resource in each glass can be characterized using a single quantifier, namely the volume of beer inside the glass. This is not true if the structure of the set of objects is more complex than in this simple example, as is the case for contextuality.

Resources theories are helpful in the investigation of many important questions, such as:
 When are the resources equivalent under free operations?
 Can we characterize the conditions on two objects such
that one of them can be converted to the other?
  Can a catalysts help?
 Can we find distillation protocols?
  What is the rate at which many copies of an object can
be converted to many copies of other object?
 Is there a resource bit, that is, an object from which all
others can be generated using free operations?
Are there different classes of resources?
 Can we create quantifiers that are related to the practical
applications of the resource? 
In the following sessions we present some of the recent developments towards an abstract unified resource theory for  contextuality and leave the treatment of  these open problems for future work.


\section{Objects and free objects}
\label{sec:comp}

Contextuality is a property exhibited by the statistics of measurements performed on a quantum system
 which shows  that such  statistics is incompatible with the description expected for classical systems. This feature  is
 closely related to the existence of \emph{incompatible measurements} in quantum systems and the compatibitity relations among measurements are the basic ingredient for contextuality. These relations  can be encoded in what we call a \emph{compatibility scenario}.

\begin{dfn}
 A \emph{compatibility scenario} is given by a triple $\Upsilon :=\left(\mathcal{M}, \mathcal{C}, \mathcal{O} \right)$, where $\mathcal{O}$ is a finite
  set, $\mathcal{M}$ is a finite set of random variables in  $\left(\mathcal{O}, \mathcal{P}\left(\mathcal{O}\right)\right)$, and 
  $\mathcal{C}$ is a family of subsets  of $\mathcal{M}$.
The elements $\boldsymbol{\gamma} \in \mathcal{C}$ are called  \emph{contexts} and the
set $\mathcal{C}$ is called the \emph{compatibility cover} of the scenario.
\end{dfn}

The random variables in  $\mathcal{M}$ represent measurements of some property of interest with possible outcomes $\mathcal{O}$ and 
each  $\boldsymbol{\gamma} \in \mathcal{C}$ consists of a  set of  properties that can be jointly accessed.
For a given context $\boldsymbol{\gamma} \in \mathcal{C}$,  the set of possible outcomes for a joint measurement of the elements of $\boldsymbol{\gamma}$ is the Cartesian product of $|\boldsymbol{\gamma}|$ copies of $\mathcal{O}$, denoted by $\mathcal{O}^{\boldsymbol{\gamma}}$.
When  the measurements in $\boldsymbol{\gamma}$ are  jointly performed, a set of outcomes $\boldsymbol{s} \in \mathcal{O}^{\boldsymbol{\gamma}}$
will be observed. This individual run of the experiment will be called a \emph{measurement event}. 

\begin{dfn}
A \emph{behavior} $\boldsymbol{B}$ for the scenario $\left(\mathcal{M}, \mathcal{C}, \mathcal{O} \right)$ is a family of probability distributions over $\mathcal{O}^{\boldsymbol{\gamma}}$, one for each context
$\boldsymbol{\gamma} \in \mathcal{C}$, that is, 
\be \boldsymbol{B} = \left\{p_{\boldsymbol{\gamma}}:  \mathcal{O}^{\boldsymbol{\gamma}}\rightarrow [0,1] \left|\sum_{\boldsymbol{s}\in \mathcal{O}^{\boldsymbol{\gamma}}} p_{\boldsymbol{\gamma}}(\boldsymbol{s})=1, \boldsymbol{\gamma} \in \mathcal{C}\right.\right\}.\ee
For each $\boldsymbol{\gamma}$, $p_{\boldsymbol{\gamma}}(\boldsymbol{s})$ gives the probability of obtaining outcomes $\boldsymbol{s}$ in a joint measurement of the elements of $\boldsymbol{\gamma}$.
\end{dfn}

Each behavior can be associated to an abstract measurement device, called a \emph{box}, with $\left|\mathcal{M}\right|$ input buttons.
The set of contexts $\mathcal{C}$ defines which input buttons can be pressed jointly. 
Each  button has a set with $\left|\mathcal{O}\right|$ associated output lights, one  of which turns on upon pressing that button.  The probabilities $p_{\boldsymbol{\gamma}}(\boldsymbol{s})$ given by the behavior $\boldsymbol{B}$ describe the internal working of the box, that is, $p_{\boldsymbol{\gamma}}(\boldsymbol{s})$ gives the probability of having the  lights $\boldsymbol{s}$ on when we press jointly the buttons in $\boldsymbol{\gamma}$.

 \begin{figure}[h!]
 \centering
 \includegraphics[scale=0.7]{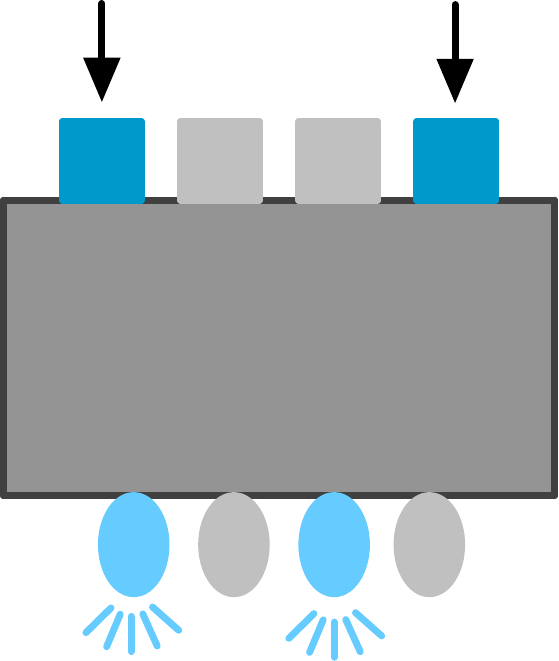}
  \caption{A box corresponding to a contextuality scenario.}
  \label{fig:box}
 \end{figure}

 In an ideal situation, it is generally assumed that behaviors must satisfy the  
\textit{non-disturbance condition}, that states that whenever two contexts $\boldsymbol{\gamma}$ and 
$\boldsymbol{\gamma}'$ overlap,
the marginal $p^{\boldsymbol{\gamma}}_{\boldsymbol{\gamma}\cap \boldsymbol{\gamma}'}$ for $\boldsymbol{\gamma} \cap \boldsymbol{\gamma}'$ computed  from  the distribution for $\boldsymbol{\gamma}$ and  the marginal $p^{\boldsymbol{\gamma}'}_{\boldsymbol{\gamma}\cap \boldsymbol{\gamma}'}$ for $\boldsymbol{\gamma} \cap \boldsymbol{\gamma}'$ computed  from  the distribution for $\boldsymbol{\gamma}'$ must coincide.

\begin{dfn}
\label{definondisturbance}
The \emph{non-disturbance} set $\mathsf{ND}\left(\Upsilon\right)$ is the set of behaviors such that for any two intersecting  contexts $\boldsymbol{\gamma}$ and $\boldsymbol{\gamma}'$
 the consistency relation $ p^{\boldsymbol{\gamma}}_{\boldsymbol{\gamma}\cap \boldsymbol{\gamma}'}= p^{\boldsymbol{\gamma}'}_{\boldsymbol{\gamma}\cap \boldsymbol{\gamma}'}$ holds. 
\end{dfn}

We
ask now if it is possible to define a distribution on
the set $\mathcal{O}^\mathcal{M}$, which specifies assignment of outcomes to all measurements, in a way that the restrictions
 yield the probabilities specified by the behavior on all
 contexts in $\mathcal{C}$ \cite{AB11}.
 
\begin{dfn}
A \emph{global section} for $\mathcal{M}$ is a probability distribution $p_{\mathcal{M}}: \mathcal{O}^\mathcal{M} \rightarrow  [0,1]$. A \emph{global section for a behaviour} 
$\boldsymbol{B} \in \mathsf{ND}\left(\Upsilon\right)$
is a global section for $\mathcal{M}$
such that the marginal probability distribution defined by $p_{\mathcal{M}}$ in each context $\boldsymbol{\gamma} \in \mathcal{C}$ is equal to $p_{\boldsymbol{\gamma}}$.  The behaviors for which there is a global section
are called \emph{non-contextual}. The set of all non-contextual behavior will be denoted by $\mathsf{NC}\left(\Upsilon\right)$
\end{dfn}

As an example, consider the scenario $\Upsilon$ containing three dicotomic measurements $\{x,y,z\}$ whith   measurement $y$  compatible with the two others. Mathematically, $\Upsilon =\left(X, 
\mathcal{C}, \mathcal{O} \right)$ with $\mathcal{O}=\left\{-1,1\right\}$,   $\mathrm{X}=\left\{x, 
y, z\right\}$
and $\mathcal{C}=\left\{\{x,y\}, \{y,z\}\right\}$. 
The graph representation of 
this scenario is  shown in Fig. \ref{fig:c3}.

 \begin{figure}[h!]
 \centering
 \includegraphics[scale=0.7]{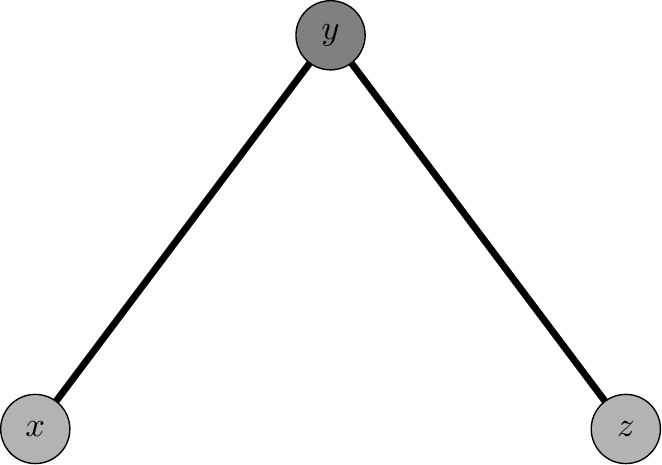}
  \caption{A compatibility scenario with three 
 measurements $x,y,z$ 
 and two contexts, $\left\{x,y\right\}$ and 
 $\left\{y,z\right\}$.}
  \label{fig:c3}
 \end{figure}
 
In this scenario, a behavior consists in specifying probability distributions
\begin{align}
 p_{xy}(ab), \,\, a,b \in \{-1,1\}, &  \ \ \ \
 p_{yz}(bc), \,\, b,c \in \{-1,1\}.
\end{align}
The non-disturbance condition demands that
\be p_y\left(b \right) \coloneqq \sum_a   p_{xy}(ab) = \sum_c 
p_{yz}(bc) \label{eq:ExampleNonDisturb}\ee
 and a behavior  is non-contextual if there is a global probability 
distribution $p_{xyz}\left(abc\right)$ such that 
\begin{align}
 p_{xy}(ab)= \sum_c p_{xyz}\left(abc\right), & \ \ \ \
 p_{yz}(bc)=\sum_a p_{xyz}\left(abc\right).
\end{align}

 The set of objects in the resource theory of contextuality considered in this contribution is the set of non-disturbing boxes and the set of free objects is the set of non-contextual boxes in arbitrary compatibility scenarios with finite set of measurements and outcomes. 


\section{Free operations}
\label{sec:free}

\subsection{Pre-processing operations}
We first consider compositions of the initial box $\boldsymbol{B}=\left(\mathcal{M},\mathcal{C},\mathcal{O}\right)$ with a 
non-contextual pre-processing box $\boldsymbol{B}_\mathrm{PRE}=\left(\mathcal{M}_{\mathrm{PRE}},\mathcal{C}_{\mathrm{PRE}},\mathcal{O}_{\mathrm{PRE}}\right)$ (see fig. \ref{fig:prepost}). We demand that the output lights of $\boldsymbol{B}_\mathrm{PRE}$ be in one-to-one correspondence with the input buttons of $\boldsymbol{B}$ in such a way that if a set of lights $\boldsymbol{r}$ in $\boldsymbol{B}_\mathrm{PRE}$ can be on jointly, the corresponding buttons $\boldsymbol{\gamma}\left(\boldsymbol{r}\right)$ in $\mathcal{M}$ form a context of $\mathcal{C}$.  This restriction guarantees that the composition  will not associate outcomes that can be produced in box $\boldsymbol{B}_\mathrm{PRE}$ with buttons that can not be jointly pressed in the box $\boldsymbol{B}$, ensuring that the transformation is consistently defined. This operation results in a box $\mathcal{W}_\mathrm{PRE}(\boldsymbol{B})$ with input buttons $\mathcal{M}_{\mathrm{PRE}}$, compatibility graph $\mathcal{C}_{\mathrm{PRE}}$, and output lights
$\mathcal{O}$. The behavior of $\mathcal{W}_\mathrm{PRE}(\boldsymbol{B})$ will be given by
\be
p_{\boldsymbol{\beta}}\left(\boldsymbol{s}\right)= \sum_{\boldsymbol{r}} p_{\boldsymbol{\gamma}(\boldsymbol{r})}(\boldsymbol{s})p_{\boldsymbol{\beta}}(\boldsymbol{r})
\ee
where $\boldsymbol{\beta}$ is any context in $\mathcal{C}_{\mathrm{PRE}}$, $\boldsymbol{r}$ runs over the possible output lights associated with context $\boldsymbol{\beta}$ in box $\boldsymbol{B}_\mathrm{PRE}$,
$\boldsymbol{\gamma}(\boldsymbol{r}) \in \mathcal{C}$ is the context of box $\boldsymbol{B}$ corresponding to $\boldsymbol{r}$ and $\boldsymbol{s}$ is one of the possible set of output lights associated with context $\boldsymbol{\gamma}(\boldsymbol{r})$.

\subsection{Post-processing operations}
We can also compose the initial box $\boldsymbol{B}$ with  a non-contextual post-processing box $\boldsymbol{B}_{\mathrm{POST}}=\left(\mathcal{M}_{\mathrm{POST}},\mathcal{C}_{\mathrm{POST}},\mathcal{O}_{\mathrm{POST}}\right)$ (see fig. \ref{fig:prepost}).  Analogously, to ensure that the transformation is consistently defined, we demand that the output lights of $\boldsymbol{B}$ be in one-to-one correspondence with the input buttons of $\boldsymbol{B}_{\mathrm{POST}}$ in such a way that if a set of lights $\boldsymbol{s}$ in $\boldsymbol{B}_\mathrm{POST}$ can be on jointly, the corresponding buttons $\boldsymbol{\delta}\left(\boldsymbol{s}\right)$ in $\mathcal{M}_{\mathrm{POST}}$ are a context of $\mathcal{C}_{\mathrm{POST}}$.
We have now a final box $\mathcal{W}_\mathrm{POST}(\boldsymbol{B})$ with input buttons $\mathcal{M}$, compatibility graph $\mathcal{C}$, and output lights
$\mathcal{O}_{\mathrm{POST}}$. The behavior of $\mathcal{W}_\mathrm{POST}(\boldsymbol{B})$ will be given by
\be
p_{\boldsymbol{\gamma}}\left(\boldsymbol{t}\right)= \sum_{\boldsymbol{s}} p_{\boldsymbol{\delta}(\boldsymbol{s})}(\boldsymbol{t})p_{\boldsymbol{\gamma}}(\boldsymbol{s})
\ee
where $\boldsymbol{\gamma}$ is any context in $\mathcal{C}$, $\boldsymbol{s}$ runs over the possible sets of output lights associated with context $\boldsymbol{\gamma}$ in box $\boldsymbol{B}$,
$\boldsymbol{\delta}(\boldsymbol{s}) \in \mathcal{C}_\mathrm{POST}$ is the context of box $\boldsymbol{B}_\mathrm{POST}$ corresponding to $\boldsymbol{s}$ and $\boldsymbol{t}$ is one of the possible output lights associated with context $\boldsymbol{\delta}(\boldsymbol{s})$.

\begin{figure}[h!]
\centering
\begin{subfigure}[t]{0.4\textwidth}
\centering
\includegraphics[scale=0.25]{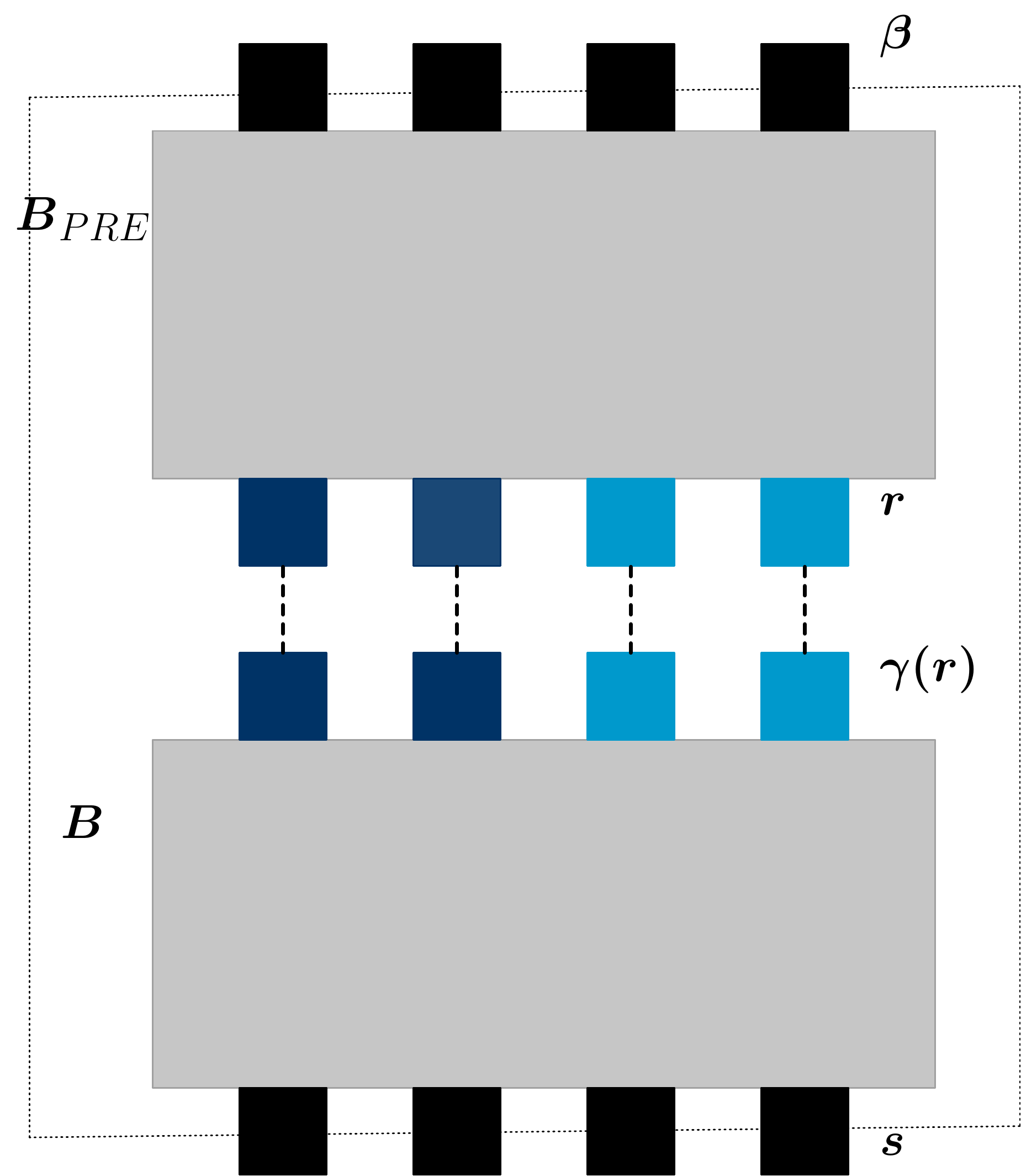}
\caption{}
\end{subfigure}
\begin{subfigure}[t]{0.4\textwidth}
\centering
\includegraphics[scale=0.25]{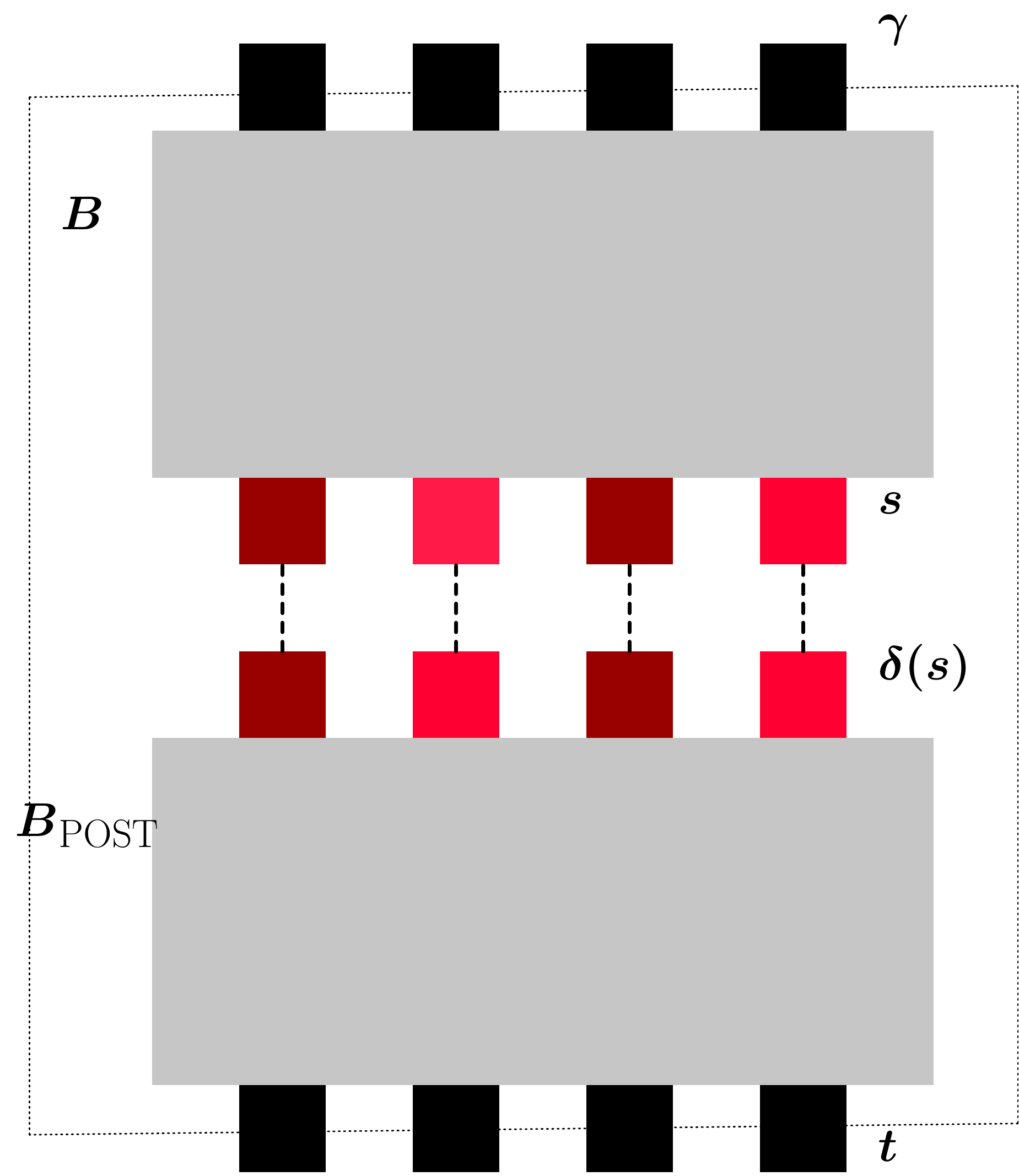}
\caption{}
\end{subfigure}
\caption{(a) A pre-processing operation that transforms the initial box $\boldsymbol{B}$ into the final box $\mathcal{W}_\mathrm{PRE}(\boldsymbol{B})$. When context $\boldsymbol{\beta}$ is chosen for box $\boldsymbol{B}_\mathrm{PRE}$ and outcome $\boldsymbol{r}$ is obtained, we have a corresponding context $\boldsymbol{\gamma}\left(\boldsymbol{r}\right)$ that is now pressed in box $\boldsymbol{B}$. If outcome $\boldsymbol{s}$ is obtained in the second box, we say that $\boldsymbol{s}$ was the outcome for context $\boldsymbol{\beta}$ in the final box. (b) A post-processing operation that transforms the initial box $\boldsymbol{B}$ into the final box $\mathcal{W}_\mathrm{POST}(\boldsymbol{B})$.  When context $\boldsymbol{\gamma}$ is chosen for box $\boldsymbol{B}$ and outcome $\boldsymbol{s}$ is obtained, we have a corresponding context $\boldsymbol{\delta}(\boldsymbol{s})$ that is now pressed in box $\boldsymbol{B}_{\mathrm{POST}}$. If outcome $\boldsymbol{t}$ is obtained in the second box, we say that $\boldsymbol{t}$ was the outcome for context $\boldsymbol{\gamma}$ in the final box. }
\label{fig:prepost}
\end{figure}

\subsection{Non-contextual Wirings}

A non-contextual wiring is a composition of a pre-processing and a post-processing operation  with the additional freedom that the behavior of the post-processing box can depend on the input $\boldsymbol{\beta}$ and output $\boldsymbol{r}$ of the pre-processing box. The behavior of the post-processing box is then of the form
$
p_{\boldsymbol{\delta}}\left(\boldsymbol{t}\vert \boldsymbol{\beta},\boldsymbol{r}\right),
$
but in such a way that each output light of the post-processing box is causally influenced only by the inputs and outputs 
of the pre-processing box that are associated with it. This  additional restriction is crucial in order not to create contextuality with the post-processing itself. To understand this restriction, consider a run of the wiring as shown in fig. \ref{fig:wirings}.

\begin{figure}[h!]
\centering
\includegraphics[scale=1.2]{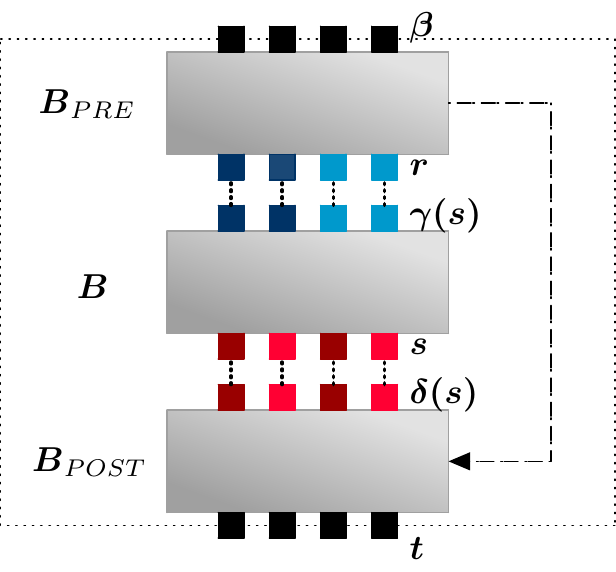}
 \caption{ A non-contextual wiring $\mathcal{W}_\mathrm{NC}$ with respect to pre- and post-processing boxes $\boldsymbol{B}_\mathrm{PRE}$ and $\boldsymbol{B}_\mathrm{POST}$, respectively,
 mapping an initial box $\boldsymbol{B}$ into a final box $\mathcal{W}_\mathrm{NC}(\boldsymbol{B})$. 
 The buttons and lights of $\mathcal{W}_\mathrm{NC}(\boldsymbol{B})$ are given by the buttons of $\boldsymbol{B}_\mathrm{PRE}$ and the lights of $\boldsymbol{B}_\mathrm{POST}$, respectively.
Only the lights (buttons) of $\boldsymbol{B}$ of the same color can be on (pressed) at the same time. 
The behavior of $\boldsymbol{B}_\mathrm{POST}$ is causally influenced by $\boldsymbol{B}_\mathrm{PRE}$, but in a restricted way such that the statistics of each output light of $\boldsymbol{B}_\mathrm{POST}$ depends only on the buttons and lights of $\boldsymbol{B}_\mathrm{PRE}$ that are associated with it. As a result, if $\boldsymbol{B}$ is non-contextual so is $\mathcal{W}_\mathrm{NC}(\boldsymbol{B})$.
 \label{fig:wirings}
}
\end{figure}

Let $\beta_i$ be an input button of $\boldsymbol{B}_{\mathrm{PRE}}$ in context $\boldsymbol{\beta}$ and $r_i$ be the associated outcome in $\boldsymbol{r}$. Let $\gamma_i$ be the input button in $\boldsymbol{B}$ associated to $r_i$ and $s_i$ be the corresponding outcome in $\boldsymbol{s}$. Let $\delta_i$ be the input button in $\boldsymbol{B}_{\mathrm{POST}}$ associated to $s_i$ and  $t_i$ be the corresponding outcome in $\boldsymbol{t}$. The consistency conditions in the definition of pre and post-processing operations guarantee that all these input buttons and output lights are well defined (see \cite{ACTA18} for details). We demand that
\be
p_{\boldsymbol{\delta}}\left(\boldsymbol{t}\vert \boldsymbol{\beta},\boldsymbol{r}\right)=\sum_{\phi} p(\phi) \prod_ip_{\delta_i\left(s_i\right)}\left(t_i\vert \beta_i, r_i\right)
\label{eq:post_rest}
\ee
where $\phi \in \Phi$ is an arbitrary additional  variable, $p(\phi)$ is a probability distribution over $\Phi$, and $p_{\delta_i}\left(t_i\vert \beta_i, r_i\right)$  is the probability of having outcome $t_i$ for button $\delta_i\left(s_i\right)$ on box $\boldsymbol{B}_{\mathrm{POST}}$ given $\phi$, $\beta_i$ and $r_i$. Fig. \ref{fig:dep} shows an example where each set of buttons and lights has exactly two elements.

\begin{figure}
  \begin{minipage}[c]{0.4\textwidth}
    \centering
\includegraphics[scale=0.4]{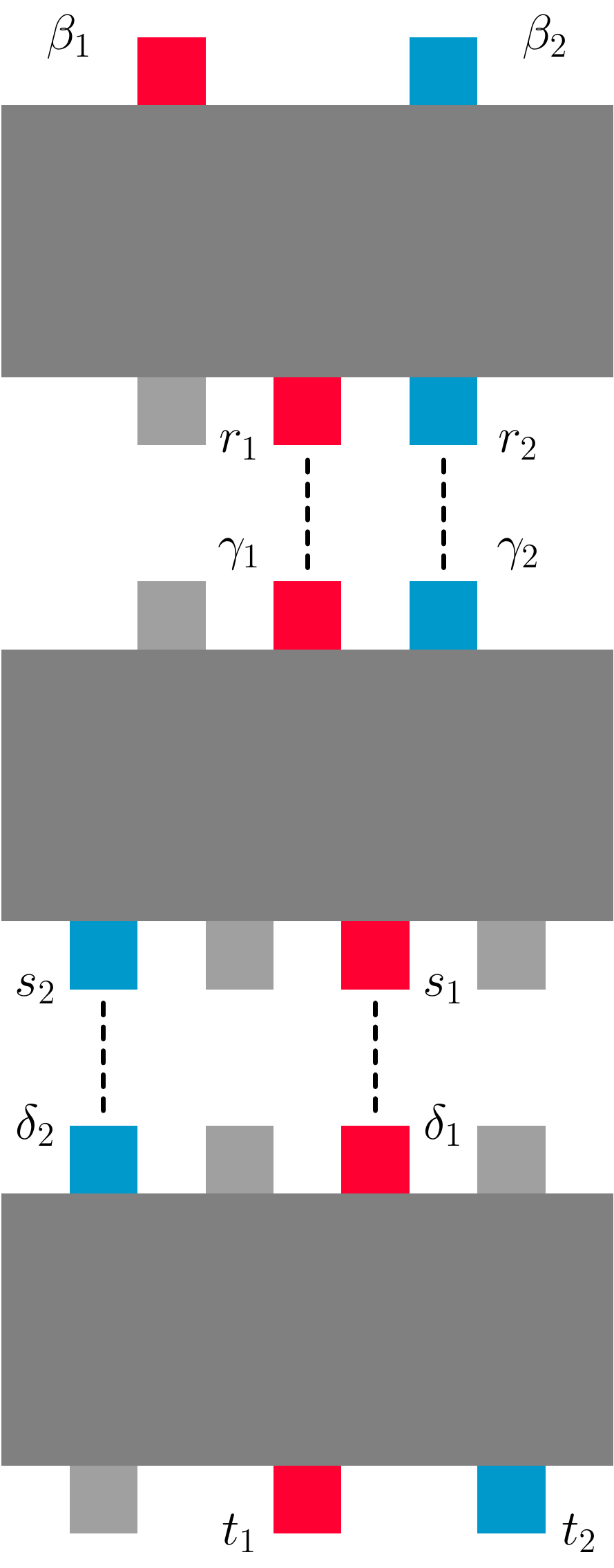}
  \end{minipage}
  \begin{minipage}[c]{0.55\textwidth}
    \caption{One round of the non-contextual wiring where context $\boldsymbol{\beta}=\left\{\beta_1, \beta_2\right\}$ is chosen for box $\boldsymbol{B}_{\mathrm{PRE}}$, leading to the sequences of buttons and lights $\left(\beta_1,r_1,\gamma_1,s_1,\delta_1,t_1\right)$ and $\left(\beta_2,r_2,\gamma_2,s_2,\delta_2,t_2\right)$. The behavior of box $\boldsymbol{B}_{\mathrm{POST}}$ can depend on box $\boldsymbol{B}_{\mathrm{PRE}}$, but with the restriction that the probability of output $t_1$ for button $\gamma_1$ can depend only on $\beta_1$ and $r_1$ and the probability of output $t_2$ for button $\delta_2$ can depend only on $\beta_2$ and $r_2$, and hence $p_{\boldsymbol{\delta}\left(\boldsymbol{s}\right)}\left(\boldsymbol{t}\vert \boldsymbol{\beta},\boldsymbol{r}\right)=p_{\delta_1}\left(t_1 \vert \beta_1, r_1\right)\times p_{\delta_2}\left(t_2 \vert \beta_2, r_2\right)$. If we consider the situation in which button $\beta_2$ is pressed only after outcome $t_1$ is recorded, this restriction implies that the behavior of the post-processing box depends only on the sequence  $\left(\beta_2,r_2,\gamma_2,s_2,\delta_2,t_2\right)$ and the  wiring has no memory of the previous round $\left(\beta_1,r_1,\gamma_1,s_1,\delta_1,t_1\right)$. We want the set of free operations to be convex and hence we also allow for convex combinations of such instances, given by the sum over the variable $\phi$ in equation \eqref{eq:post_rest}.}
 \label{fig:dep}
  \end{minipage}
\end{figure}

This composition defines   a final box $\mathcal{W}_\mathrm{NC}\left(\boldsymbol{B}\right)$ with input buttons $\mathcal{M}_{\mathrm{PRE}}$, compatibility graph $\mathcal{C}_{\mathrm{PRE}}$, and output lights
$\mathcal{O}_{\mathrm{POST}}$. The behavior of the final box will be given by
\be
p_{\boldsymbol{\beta}}\left(\boldsymbol{t}\right)= \sum_{\boldsymbol{s},\boldsymbol{r}} p_{\boldsymbol{\delta}(\boldsymbol{s})}(\boldsymbol{t})p_{\boldsymbol{\gamma}(\boldsymbol{r})}(\boldsymbol{s})p_{\boldsymbol{\beta}}(\boldsymbol{r})
\ee
where $\boldsymbol{\beta} \in \mathcal{C}_{\mathrm{PRE}}$, $\boldsymbol{r}$ runs over the possible output lights associated with context $\boldsymbol{\beta}$ in box $\boldsymbol{B}_{\mathrm{PRE}}$,
$\boldsymbol{\gamma}(\boldsymbol{r}) \in \mathcal{C}$ is the context of box $\boldsymbol{B}$ corresponding to $\boldsymbol{r}$, $\boldsymbol{s}$ runs over the output lights associated with context $\boldsymbol{\gamma}(\boldsymbol{r})$, $\boldsymbol{\delta}(\boldsymbol{s})\in \mathcal{C}_{\mathrm{POST}}$ is the context of box $\boldsymbol{B}_{\mathrm{POST}}$ corresponding to $\boldsymbol{s}$, and $\boldsymbol{t}$ is one of the possible output lights associated with context $\boldsymbol{\delta}(\boldsymbol{s})$.

The set of all non-contextual wirings will be denoted by $\mathsf{NCW}$. Self-consistency of the theory requires that non-contextual wirings satisfy the following property, proven here for context with two buttons. The general case is completely analogous and a general proof can be found in the supplemental material of ref. \cite{ACTA18}.
\begin{lem}[Non-disturbance preservation] 
\label{Lem:ND}
The class of boxes $\mathsf{ND}$ is closed under all wirings in $\mathsf{NCW}$.
\end{lem}

\begin{proof}
 Suppose that  $\boldsymbol{\beta}=\left\{\beta_1,\beta_2\right\}$. Since when one button is pressed exactly one light turns on, we have $\boldsymbol{r}=\left\{r_1,r_2\right\}$, $\boldsymbol{\gamma}\left(\boldsymbol{r}\right)=\left\{\gamma_1\left(r_1\right),\gamma_2\left(r_2\right)\right\}$, $\boldsymbol{s}=\left\{s_1,s_2\right\}$, $\boldsymbol{\delta}\left(\boldsymbol{s}\right)=\left\{\delta_1\left(s_1\right),\delta_2\left(s_2\right)\right\}$, as in fig. \ref{fig:dep}.
 We have that 
 \begin{eqnarray}
  \sum_{t_2}p_{\boldsymbol{\beta}}\left(\boldsymbol{t}\right)&=&\sum_{t_2} \sum_{\boldsymbol{r},\boldsymbol{s}} p_{\boldsymbol{\delta}\left(\boldsymbol{s}\right)}\left(\boldsymbol{t}\right)p_{\boldsymbol{\gamma}\left(\boldsymbol{r}\right)}\left(\boldsymbol{s}\right)p_{\boldsymbol{\beta}}\left(\boldsymbol{r}\right)  \\
  &=& \sum_{t_2} \sum_{\boldsymbol{r},\boldsymbol{s}} \sum_{\phi}p\left(\phi\right)p_{\delta_1\left(s_1\right)}\left(t_1\vert \beta_1, r_1, \phi\right)p_{\delta_2\left(s_2\right)}\left(t_2\vert \beta_2, r_2, \phi\right)p_{\boldsymbol{\gamma}\left(\boldsymbol{r}\right)}\left(\boldsymbol{s}\right)p_{\boldsymbol{\beta}}\left(\boldsymbol{r}\right)  \\
  &=&\sum_{\boldsymbol{r},\boldsymbol{s}} \sum_{\phi}p\left(\phi\right)p_{\delta_1\left(s_1\right)}\left(t_1\vert \beta_1, r_1, \phi\right)\left[\sum_{t_2} p_{\delta_2\left(s_2\right)}\left(t_2\vert \beta_2, r_2, \phi\right)\right]p_{\boldsymbol{\gamma}\left(\boldsymbol{r}\right)}\left(\boldsymbol{s}\right)p_{\boldsymbol{\beta}}\left(\boldsymbol{r}\right)  \\
   &=&\sum_{r_1,r_2,s_1,s_2} \sum_{\phi}p\left(\phi\right)p_{\delta_1\left(s_1\right)}\left(t_1\vert \beta_1, r_1, \phi\right)p_{\boldsymbol{\gamma}\left(\boldsymbol{r}\right)}\left(\boldsymbol{s}\right)p_{\boldsymbol{\beta}}\left(\boldsymbol{r}\right)  \\
   &=&\sum_{r_1,s_1} \sum_{\phi}p\left(\phi\right)p_{\delta_1\left(s_1\right)}\left(t_1\vert \beta_1, r_1, \phi\right)\sum_{r_2,s_2} p_{\boldsymbol{\gamma}\left(\boldsymbol{r}\right)}\left(\boldsymbol{s}\right)p_{\boldsymbol{\beta}}\left(\boldsymbol{r}\right)  \\
   &=&\sum_{r_1,s_1} \sum_{\phi}p\left(\phi\right)p_{\delta_1\left(s_1\right)}\left(t_1\vert \beta_1, r_1, \phi\right) p_{\gamma_1(r_1)}(s_1)\sum_{r_2}p_{\boldsymbol{\beta}}\left(\boldsymbol{r}\right)  \\
   &=&\sum_{r_1,s_1} \sum_{\phi}p\left(\phi\right)p_{\delta_1\left(s_1\right)}\left(t_1\vert \beta_1, r_1, \phi\right) p_{\gamma_1(r_1)}(s_1)p_{\beta_1}\left(r_1\right)\\
   &=&p_{\beta_1}\left(t_1\right).
 \end{eqnarray}
\end{proof}

In addition, to give valid free operations, $\mathsf{NCW}$ must fulfill the following requirement, proven here for contexts with two buttons \cite{ACTA18}.
\begin{thm}[Non-contextuality preservation]
\label{teoncpreservation}
The class of boxes $\mathsf{NC}$ is closed under all wirings in $\mathsf{NCW}$.
 \end{thm}
 
 \begin{proof}
  Non-conextuality of $\boldsymbol{B}_{\mathrm{PRE}}$, $\boldsymbol{B}$ and $B_{\mathrm{POST}}$  and condition \eqref{eq:post_rest} imply that 
  \begin{align}
   p_{\boldsymbol{\beta}}\left(\boldsymbol{r}\right)&=\sum_{\xi}p\left(\xi\right)p_{\beta_1}\left(r_1\vert \xi\right)p_{\beta_2}\left(r_2\vert \xi\right),\\
   p_{\boldsymbol{\gamma}\left(\boldsymbol{r}\right)}\left(\boldsymbol{s}\right)&=\sum_{\psi}p\left(\psi\right)p_{\gamma_1\left(r_1\right)}\left(s_1\vert \psi\right)p_{\gamma_2\left(r_2\right)}\left(s_2\vert \psi\right),\\
   p_{\boldsymbol{\delta}\left(\boldsymbol{s}\right)}\left(\boldsymbol{t}\right)&=\sum_{\phi}p\left(\phi\right)p_{\delta_1\left(s_1\right)}\left(t_1\vert \beta_1,r_1, \phi\right)p_{\delta_2\left(r_2\right)}\left(t_2\vert \beta_2,r_2,\phi\right).
  \end{align}
  This in turn implies that 
  \begin{align}
p_{\boldsymbol{\beta}}\left(\boldsymbol{t}\right) &= \sum_{\boldsymbol{r},\boldsymbol{s}} p_{\boldsymbol{\delta}\left(\boldsymbol{s}\right)}\left(\boldsymbol{t}\right)p_{\boldsymbol{\gamma}\left(\boldsymbol{r}\right)}\left(\boldsymbol{s}\right)p_{\boldsymbol{\beta}}\left(\boldsymbol{r}\right)  \\
&=\sum_{\xi,\psi,\phi}\sum_{r_1,r_2,s_1,s_2}p\left(\xi\right)p\left(\psi\right)p\left(\phi\right)p_{\delta_1\left(s_1\right)}\left(t_1\vert \beta_1,r_1, \phi\right)p_{\delta_2\left(r_2\right)}\left(t_2\vert \beta_2,r_2, \phi\right) \nonumber\\
& \hspace{9em} \times p_{\gamma_1\left(r_1\right)}\left(s_1\vert \psi\right)p_{\gamma_2\left(r_2\right)}\left(s_2\vert \psi\right)p_{\beta_1}\left(r_1\vert \xi\right)p_{\beta_2}\left(r_2\vert \xi\right)\\
&= \sum_{\boldsymbol{\psi}} p\left(\boldsymbol{\psi}\right)\left[ \sum_{r_1,s_1}p_{\delta_1\left(s_1\right)}\left(t_1\vert \beta_1,r_1, \phi\right)p_{\gamma_1\left(r_1\right)}\left(s_1\vert \psi\right)p_{\beta_1}\left(r_1\vert \xi\right)\right]\nonumber\\
& \hspace{9em} \times\left[\sum_{r_2,s_2}p_{\delta_2\left(r_2\right)}\left(t_2\vert \beta_2,r_2,\phi\right)p_{\gamma_2\left(r_2\right)}\left(s_2\vert \xi\right)p_{\beta_2}\left(r_2\vert \xi\right)\right]\\
&=\sum_{\boldsymbol{\psi}} p\left(\boldsymbol{\psi}\right)p_{\beta_1}\left(t_1\vert\boldsymbol{\psi}\right)p_{\beta_2}\left(t_2\vert\boldsymbol{\psi}\right).
  \end{align}
  where $\boldsymbol{\psi}=\left(\xi,\psi,\phi\right).$
 \end{proof}

This proof is connected to the fact that the composition of any three independent 
non-contextual boxes yields a final box that is also non-contextual (with three independent non-contextual hidden variables). $\mathsf{NCW}$ is however more powerful than such compositions because the pre- and post-processing boxes here are not independent. Still, the restriction of eq.\ \eqref{eq:post_rest} enables non-contextuality preservation (see \cite{ACTA18}).
For space-like separated measurements, $\mathsf{NCW}$ reduces to \emph{local operations assisted by shared randomness}, the canonical free operations of Bell non-locality \cite{GWAN12,Vicente14,GA17}. This also shows that non-contextual wirings is not the largest set of free operations for contextuality, since in the particular case of Bell scenarios it is known that local operations assisted by shared randomness is not the largest set of free operations of non-locality. However, we still lack an explicit parametrization for a larger set of free operations for contextuality, and we restrict throughout to the class of non-contextual wirings, unless stated otherwise, since this is the class for which we have a friendly parametrization with a clear physical interpretation.

\subsubsection*{Product and controlled choice of boxes}

We consider now two different ways of combining independent boxes $\boldsymbol{B}_1=\left(\mathcal{M}_1, \mathcal{C}_1, \mathcal{O}_1\right)$ and $\boldsymbol{B}_2=\left(\mathcal{M}_2, \mathcal{C}_2, \mathcal{O}_2\right)$. 
First we define the box $\boldsymbol{B}_1 \otimes \boldsymbol{B}_2$, called the \emph{product} of  $\boldsymbol{B}_1$ and $\boldsymbol{B}_2$, as the box such that each of  its  contexts  is given by  $\boldsymbol{\gamma}=\boldsymbol{\gamma}_1  \cup \boldsymbol{\gamma}_2$, with $\boldsymbol{\gamma}_i\in \mathcal{C}_i$, that is, each context in the final box $\boldsymbol{B}_1 \otimes \boldsymbol{B}_2$ consists of a choice of context for box $\boldsymbol{B}_1$ \emph{and} a choice of context for $\boldsymbol{B}_2$ . The behavior of this box is
\be p_{\boldsymbol{\gamma}_1\cup \boldsymbol{\gamma}_2}\left(\boldsymbol{s}_1,\boldsymbol{s}_2\right)=p_{\boldsymbol{\gamma}_1}\left(\boldsymbol{s}_1\right) p_{\boldsymbol{\gamma}_2}\left(\boldsymbol{s}_2\right).\ee

\begin{figure}[h!]
\centering
\includegraphics[scale=0.7]{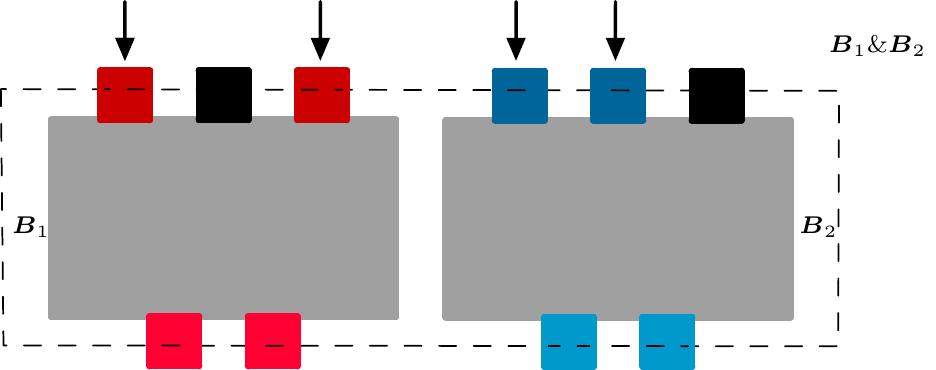}
 \caption{The box  $\boldsymbol{B}_1 \otimes \boldsymbol{B}_2$ for which each context  consists of a choice of context for box $\boldsymbol{B}_1$ \emph{and} a choice of context for $\boldsymbol{B}_2$.}
 \label{fig:and}
\end{figure}

We can also define the box $\boldsymbol{B}_1 \& \boldsymbol{B}_2$, called the \emph{controlled choice} of  $\boldsymbol{B}_1$ and $\boldsymbol{B}_2$,
as the box such that  $\mathcal{C}=\mathcal{C}_1  \cup \mathcal{C}_2$, that is, each context in the final box $\boldsymbol{B}_1 \& \boldsymbol{B}_2$ consists of a choice of context for box $\boldsymbol{B}_1$ \emph{or} a choice of context for $\boldsymbol{B}_2$. The behavior of this box is
 a juxtaposition of a behavior for box $\boldsymbol{B}_1$ and a behavior for $\boldsymbol{B}_2$.

\begin{figure}[h!]
\centering
\begin{subfigure}[b]{0.45\textwidth}
\includegraphics[scale=0.7]{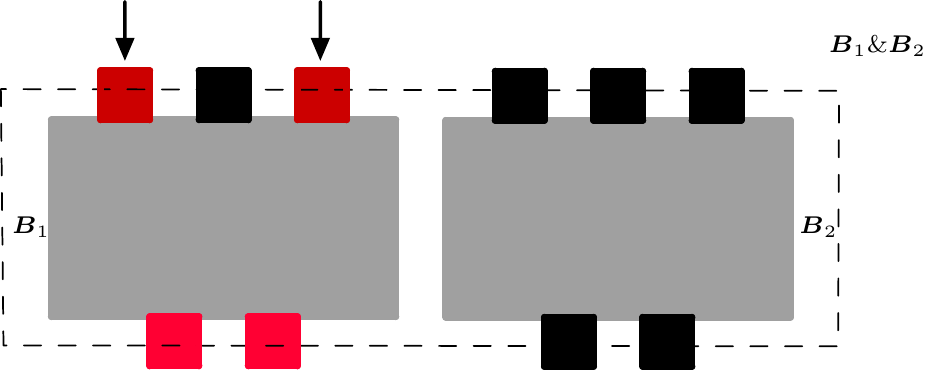}
\caption{}
\end{subfigure}
~ \qquad 
\begin{subfigure}[b]{0.45\textwidth}
\includegraphics[scale=0.7]{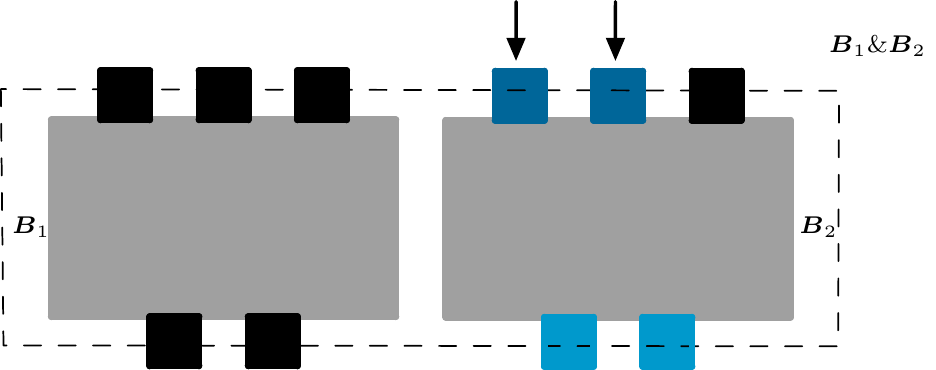}
\caption{}
\end{subfigure}
\caption{The box  $\boldsymbol{B}_1 \& \boldsymbol{B}_2$ for which each context  consists of a choice of context for box $\boldsymbol{B}_1$ \emph{or} a choice of context for $\boldsymbol{B}_2$.}
\label{fig:or} 
\end{figure}

 \section{Quantifiers}
 \label{sec:quant}
The essential requirement for a function to be a valid measure of contextuality is that it is monotonous (i.e. non-increasing) under the set of non-contextual wirings. 

\begin{dfn}
A function $Q: \mathsf{ND}\left(\Upsilon\right)\rightarrow \mathds{R}$ is a \emph{contextuality monotone} for the resource theory of contextuality defined by non-contextual wirings if
 \be Q\left[\mathcal{W}\left(B\right)\right] \leq Q\left(B\right)\ee
 for every $\mathcal{W} \in \mathsf{NWC}$.
 \end{dfn}

Besides monotonicity under free operations, other properties of a monotone $Q$ are also desirable \cite{HGJKL15, ABM17}:

\begin{enumerate}

\item  \emph{Faithfullness:} For all $\boldsymbol{B}\in \mathsf{NC}(\Upsilon)$, $Q\left(\boldsymbol{B}\right)=0$.

\item  \emph{Preservation under reversible operations:} If $\mathcal{T} \in \mathcal{W}$ is reversible, then
\be Q\left(\mathcal{T}\left(\boldsymbol{B}\right)\right) = Q\left(\boldsymbol{B}\right).\ee

\item  \emph{Additivity:} Given two independent boxes $\boldsymbol{B}_1$ and $\boldsymbol{B}_2$ we require: 
\be
Q\left(\boldsymbol{B}_1 \otimes \boldsymbol{B}_2\right)\leq Q\left(\boldsymbol{B}_1\right)+Q\left(\boldsymbol{B}_2\right).
\ee
\be
Q\left(\boldsymbol{B}_1 \& \boldsymbol{B}_2\right)\leq Q\left(\boldsymbol{B}_1\right)+Q\left(\boldsymbol{B}_2\right).
\ee

\item  \emph{Convexity:} If a behavior can be written as $\boldsymbol{B}=\sum_i \pi_i \boldsymbol{B}^i$, where $\pi_i \in [0,1]$ and each $\boldsymbol{B}^i$ is a behavior  for the same scenario, then 
\be Q\left(\boldsymbol{B}\right) \leq \sum_i \pi_i Q\left(\boldsymbol{B}^i\right).\ee

\item  \emph{Continuity:} $Q\left(\boldsymbol{B}\right)$ should be a continuous function of $\boldsymbol{B}$.

\end{enumerate}
 
In what follows we exhibit a number of monotones for different resource theories of contextuality and
 list which of the  properties above they satisfy.
 
 \subsection{Entropic Contextuality Quantifiers}

 \subsubsection{Relative entropy of contextuality}
 
  In ref. \cite{GHHHHJKW14}, the authors also introduce two measures of contextuality based directly on the notion of relative entropy distance, 
 also called the Kullback-Leibler divergence. Given two probability distributions $p$ and $q$ in a sample space $\Omega$, the 
 Kullback-Leiber divergence between $p$ and $q$ 
 \be D_{\mathrm{KL}}(p\|q) = \sum_{i\in \Omega} p(i) \, \log\frac{p(i)}{q(i)} \ee
 is a measure of  the difference between the two probability distributions.
 
 \begin{dfn}
 The \emph{Relative Entropy of 
 Contextuality} of a behavior $\boldsymbol{B}$ is defined as
 \begin{equation}\label{HorDist2}
E_{max}\left(\boldsymbol{B} \right) =   \min _{\boldsymbol{B}^{NC} \in \mathsf{NC} }\ \ \max _{\pi}\ \ 
\sum _{\boldsymbol{\gamma}\in \mathcal{C}}\ \pi\left(\boldsymbol{\gamma}\right) D_{\mathrm{KL}}\left(p_{\boldsymbol{\gamma}}  \middle\| p^{NC}_{\boldsymbol{\gamma}} \right),
\end{equation}
where the minimum is taken over all non-contextual behaviors $\boldsymbol{B}^{NC}=\left\{ p^{NC}_{\boldsymbol{\gamma}} \right\}$ and the maximum is taken over all
 probability distributions $\pi$ defined on the set of contexts $\mathcal{C}$.
The \emph{Uniform Relative Entropy of Contextuality} of $\boldsymbol{B}$
is defined as
\begin{equation}\label{HorDist3}
E_{u}\left(\boldsymbol{B} \right) = \frac{1}{N} \min _{\boldsymbol{B}^{NC}\in \mathsf{NC} } \ \ \sum _{\boldsymbol{\gamma}\in \mathcal{C}}\ D_{\mathrm{KL}}\left(p_{\boldsymbol{\gamma}}  \middle\| p^{NC}_{\boldsymbol{\gamma}} \right),
\end{equation}
 where $N=\left|\mathcal{C}\right|$ is the number of contexts in $\mathcal{C}$ and, once more, the minimum is taken over all non-contextual behaviors $\boldsymbol{B}^{NC}=\left\{ p^{NC}_{\boldsymbol{\gamma}} \right\}$.
 \end{dfn}

In reference \cite{ACTA18} it is shown that $E_{max}$ is a monotone under \emph{non-contextual wirings}. The quantity $E_{u}$, however, is not a monotone under the complete class of 
non-contextual wirings, as shown in Ref. \cite{GA17} for the special class of 
Bell scenarios. Nonetheless, it is a monotone under a broad class of such operations. More specifically, it is monotone under
post-processing operations and under a subclass of pre-processing operations (see ref. \cite{AT17}).

\begin{thm}
\label{teo:prop_entropic}
The following properties are valid for the contextuality quantifiers based on relative entropy:
\begin{enumerate}
\item $E_{max}$ is a contextuality monotone for the resource theory of contextuality defined by  non-contextual wirings;
\item $E_u$ is a contextuality monotone for the resource theory of contextuality defined by post-processing operations and a subclass of pre-processing operations;
\item $E_{max}$ and $E_u$ are faithful, additive, convex, continuous, and preserved under relabellings of inputs and outputs.
\end{enumerate}

\end{thm}

 The proof of this result can be found in refs. \cite{GHHHHJKW14,HGJKL15,AT17}.
 
 \subsection{Geometric Contextuality Quantifiers}
 \label{sectiongeometric}

We now introduce contextuality monotones  based on the distance $\ell_1$, in contrast 
with the previous defined quantifiers which are based on entropic distances.

\begin{dfn}
The $\ell_1$-\emph{max contextuality distance} of a behavior $\boldsymbol{B}$ is defined as
\be \mathcal{D}_{max}\left(\boldsymbol{B}\right)= \min _{\boldsymbol{B}^{NC}\in \mathsf{NC} } \max_{\pi} \sum _{\boldsymbol{\gamma}\in \mathcal{C}}\ \pi\left(\boldsymbol{\gamma}\right)\ \sum_{\boldsymbol{s}}\left|p_{\boldsymbol{\gamma}}\left(\boldsymbol{s}\right)-p^{NC}_{\boldsymbol{\gamma}}\left(\boldsymbol{s}\right) \right|,  \label{eqdefdist3}\ee
where the minimum is taken over all non-contextual behaviors $\boldsymbol{B}^{NC}=\left\{ p^{NC}_{\boldsymbol{\gamma}} \right\}$ and the maximum is taken over all
over all probability distributions $\pi$ defined over the set of contexts $\mathcal{C}$.
The $\ell_1$-\emph{uniform contextuality distance} of a behavior $B$ is defined as
\be \mathcal{D}_u\left(\boldsymbol{B}\right)=\frac{1}{N}\min_{\boldsymbol{B}^{NC}\in \mathsf{NC} } \sum_{\boldsymbol{\gamma}\in \mathcal{C}} \sum_{\boldsymbol{s}}\left|p_{\boldsymbol{\gamma}}\left(\boldsymbol{s}\right)-p^{NC}_{\boldsymbol{\gamma}}\left(\boldsymbol{s}\right) \right|, \label{eqdefdist2}\ee
where $N=\left|\mathcal{C}\right|$ is the number of contexts in $\mathcal{C}$.
\end{dfn}

For detailed discussion of these contextuality quantifiers see ref. \cite{AT17} and, for the special class of
Bell scenarios, ref. \cite{BAC18}).

\begin{thm}
\label{teo_dist}
The following properties are satisfied:
\begin{enumerate}
\item $\mathcal{D}_{max}$ is a contextuality monotone for the resource theory of contextuality defined by the non-contextual wiring operations;
\item $\mathcal{D}_u$ is a contextuality monotone for the resource theory of contextuality defined by post-processing operations and a subclass of pre-processing operations;
\item $\mathcal{D}_{u}$ and $\mathcal{D}_{max}$ are faithful, additive, convex, continuous, and preserved under relabellings of inputs and outputs.
\item $\mathcal{D}_u$ can be computed using linear programming.
\end{enumerate}
\end{thm}

This result is proven in refs. \cite{AT17,BAC18}. It shows that while $\mathcal{D}_{max}$ is a proper 
contextuality monotone under the entire  class of non-contextual wirings,
 $D_u$ are more suitable when the set of allowed free operations preserves the scenario under consideration.  Other distances defined in the set $\mathsf{ND}$ can also be used in place of the $\ell_1$ distance. The above results are also valid for any $\ell_p$ distance.

\subsection{Contextual Fraction}
\label{subsec:cf}

A contextuality quantifier based on the intuitive notion of what \emph{fraction} of a given behavior admits a non-contextual description
was introduced in refs. \cite{AB11, ADLPBC12}. Several properties of this quantifier were further discussed in Ref. \cite{ABM17}.

\begin{dfn}
The \emph{contextual fraction} of a behavior $\boldsymbol{B}$ is defined as
\be
\label{eq:cont_frac}
\mathcal{CF}\left(\boldsymbol{B}\right)= \min \left\{\lambda \left|\boldsymbol{B}=  \lambda \boldsymbol{B}' + \left(1-\lambda\right)\boldsymbol{B}^{NC}\right.\right\},
\ee
where the minimum is taken over all decompositions of $\boldsymbol{B}$ as a convex sum of a non-contextual behavior  $\boldsymbol{B}^{NC}$ and  an arbitrary  behavior $\boldsymbol{B}'$.
\end{dfn}

\begin{thm}
 The contextual fraction is a monotone under all linear operations that preserve the non-contextual set $\mathsf{NC}$.
\end{thm}

\begin{proof}Let $\mathcal{T}$ be a linear operation over the set of behaviors such that 
\be
\mathcal{T}\left(\mathsf{NC}\right) \subset \mathsf{NC}.
\ee
Given a behavior $\boldsymbol{B}$, let $\boldsymbol{B}= \lambda \boldsymbol{B}' + \left(1-\lambda\right)\boldsymbol{B}^{NC}$ be the decomposition of $\boldsymbol{B}$ achieving the 
minimum in eq. \eqref{eq:cont_frac}, that is, $\mathcal{CF}\left(\boldsymbol{B}\right)=\lambda$. Then
\begin{eqnarray}
 \mathcal{T}\left(\boldsymbol{B}\right) &=&\mathcal{T}\left(\lambda \boldsymbol{B}' + \left(1-\lambda\right)\boldsymbol{B}^{NC}\right)\\
 &=&\lambda\mathcal{T}\left( \boldsymbol{B}' \right) + \left(1-\lambda\right)\mathcal{T}\left(\boldsymbol{B}^{NC}\right).
\end{eqnarray}
Since $\mathcal{T}\left(\boldsymbol{B}^{NC}\right)$ is a non-contextual behavior, we conclude that
\be
\mathcal{CF}\left(\mathcal{T}\left(\boldsymbol{B}\right)\right) \leq \lambda= \mathcal{CF}\left(\boldsymbol{B}\right).
\ee
\end{proof}

\begin{prop}
 \label{teo:cf}
 The contextual fraction  satisfies:
 \begin{enumerate}
  \item The contextual fraction is faithful, convex and continuous;
  \item $\mathcal{CF}\left(\boldsymbol{B}_1 \& \boldsymbol{B}_2\right) \leq \max_i \mathcal{CF}\left(\boldsymbol{B}_i\right)$;
  \item $\mathcal{CF}\left(\boldsymbol{B}_1 \otimes  \boldsymbol{B}_2\right) \leq \mathcal{CF}\left(\boldsymbol{B}_1\right) + \mathcal{CF}\left(\boldsymbol{B}_2\right) - \mathcal{CF}\left(\boldsymbol{B}_1\right)\mathcal{CF}\left(\boldsymbol{B}_2\right)$;
 \item The contextual fraction can be calculated via linear programming.
 \end{enumerate}

\end{prop}

The proof of these results can be found in Ref. \cite{ABM17}.

\section{Contextuality as a resource}
\label{sec:app}

In this contribution we have defined the set of free objects using an abstract mathematical characterization of noncontextuality, but a resource theory of contextuality will exhibit its true power when applied to operational applications of this phenomenon. Contextuality has been identifyed as a possible resource for quantum advantages in different schemes and in this section we review some of the recent results.

\subsubsection*{Contextuality and random number generation}

The  generation  of  genuine  randomness is still a challenging task as true random  numbers  can  never  be  generated   with classical systems, for which   a  deterministic description, in principle,  always exists. For quantum system that exhibit  contextuality, a deterministic description that is independent on the choice of measurement settings is impossible, thus opening the door for the generation of genuine random numbers. This was indeed achieved in refs. \cite{PAMGMMOHLMM10,UZZWYDDK13} where the violation of  Bell and non-contextuality inequalities where used directly to compute a lower bound on the min-entropy of the outcomes, thus guaranteeing randomness of the string of output bits. This string can then be processed using classical algorithms to distill genuine random numbers. It is interesting to note that the quantum system in ref. \cite{UZZWYDDK13} is a qutrit, which shows that randomness can be generated without the need of using costly quantum resources such as entanglement. This  allows for easier implementation and significantly higher generation rate of random strings.

\subsubsection*{Contextuality and models of quantum computation with state injection}

Quantum computation with state injection (QCSI) \cite{BK05} is a scheme composed of a free part consisting of  quantum  circuits with restricted set of states, unitaries and measurements (generally  restricted to be that of  the stabilizer formalism) in which quantum computation universality is achieved by the injection  of   special  resource  states, called \emph{magic states}. These special states   are usually  distilled  from  many  copies  of  noisy  states  through a procedure called \emph{magic  state  distillation}.

To understand the source of quantum advantage in these schemes we need to understand what is precisely the quantum property that allows for magic state distillation. In ref. \cite{VFGE12}, it was shown that for a special choice of Wigner function representation of qudits with prime $d$, its positivity  implied the existence of an efficient classical simulation of the state, which in turn implied that this state is not useful for magic state distillation. This result was later explored in ref. \cite{HWVE14}, which exhibited  a contextuality scenario based on the set of restricted measurements for which non-contextuality is equivalent to  negativity of the Wigner function. This shows that contextuality with respect to this scenario is a necessary ingredient for magic state distillation. 

This result does not easily generalizes to quibt systems \cite{BDBOR17,LWE18}, where both the definition of the Wigner function and the presence of state independent contextuality with respect to the restricted measurements poses an obstacle to the 
recognition of contextuality as a resource. Several attempts to establish contextuality as a resource in qubit stabilizer sub-theory have  been done so by further restricting to non-contextual subsets of operations within the qubit stabilizer sub-theory. In ref. \cite{DGBR15}, the authors restrict to qubits with real density matrices (rebits) and define a Wigner function for $n$ rebits that is consistent with the restricted stabilizer formalism. With this construction they are able to prove that there is a real QCSI schemes in which universal quantum computation is only possible in the presence of contextuality. 
In ref. \cite{RBDOB17}, the authors show that if non-negative Wigner functions remain non-negative
under free measurements, then contextuality and Wigner function negativity are necessary resources for universal quantum computation on these schemes. The result on contextuality
is however strictly stronger than the result
on Wigner functions, since different from the qudit case \cite{DOBBR17}, qubit magic states can have negative Wigner functions but
still be non-contextual.
These results where later generalized in ref. \cite{BDBOR17}, that shows that if the set of available measurements in the scheme is such that there exists a quantum states that does not exhibit contextuality, then contextuality is a necessary resource for universal quantum computation on these schemes.

\subsubsection*{Contextuality and measurement based quantum computation}

A $\ell_d$-measurement-based quantum computation ($\ell_d$-MBQC) \cite{RBB03} consists of a $n$-site correlated resource state and a control classical computer with restricted computational power. Each site receives the information of a measurement setting to be performed in its system, encoded as an element of $\mathds{Z}_d$, with $d=p^r$ and $p$ prime,  and  returns the outcome of the measurement also encoded as an element of $\mathds{Z}_d$. No communication between sites is allowed during the computation. The control computer  post-processes the measurement outcomes linearly to produce the output of the computation.

For $d=2$ it was shown in ref. \cite{AB09} that nonlinear Boolean functions can be computed with  $\ell_2$-MBQC  with a resource state constructed from a proof of contextuality based on Mermin's GHZ paradox. It was then shown in ref. \cite{Raussendorf13} that deterministic computation of any nonlinear Boolean function with $\ell_2$-MBQC implies that the contextual fraction of the corresponding behavior is equal to one. For probabilistic computation,  $\ell_2$-MBQC which compute a non-linear Boolean function with high probability are necessarily contextual. Ref. \cite{OG17} proves that  bipartite non-local behaviors in the CHSH scenario and behaviors with arbitrarily small violation of a multi-partite GHZ non-contextuality inequality suffice for reliable classical computation, that is, for the evaluation of any Boolean function with success probability bounded away from $\frac{1}{2}$.

These results were later connected with the contextual fraction of the resource state with respect to the available measurements in the computation \cite{ABM17}. Let $f$ be  a  Boolean  function and consider an $\ell_2$-MBQC that uses the behavior $\boldsymbol{B}$ to   compute $f$ with  average  success  probability $p_S$ overall possible inputs, and corresponding average failure probability $p_F=1-p_S$.  Then
\be p_F\geq \left(1-\mathcal{CF}\left(\boldsymbol{B}\right)\right)\nu\left(f\right), \ee
where $\nu\left(f\right)$ is the average distance of $f$ to the closest $\mathds{Z}_2$-linear function\footnote{The average distance between  two  Boolean  functions $f,g:2^m\rightarrow 2^l$ is   given   by $d(f,g):=\frac{1}{2^{m}}\left|\left\{i\in 2^m\vert f(i)\neq g(i)\right\}\right|.$}.

In the qudit case, however, examples of non-contextual $\ell_d$-MBQCs with local dimension $d\geq 3$ that evaluate nonlinear functions can be found \cite{FRB18}. Nevertheless, it is still possible to connect contextuality with quantum advantages. In this direction, ref. \cite{HWB11} shows that the evaluation of a sufficiently high order polynomial function on a multi-qudit system provides  a proof of contextuality. This problem was also investigated in ref. \cite{FRB18}, that besides reproducing the result of \cite{HWB11},  emphasised the distinctive role of contextuality in individual sites versus  strong correlations between the sites.

\subsubsection*{Memory cost of simulating contextuality}

Contextuality  can be simulated by classical models with memory and  the efficiency of such simulations can help understand the difference between quantum and classical systems. In any such simulation, the system changes between  different internal states during the measurement sequence. These states, drawn from a classical state space $\Lambda$,  can be considered as memory  and define the spatial complexity of the simulation.  The model in which the cardinality of $\Lambda$ is minimum is memory-optimal and defines the \emph{memory cost} of the simulation. 

In refs. \cite{KGPLC11,FK17} the authors study the memory cost of simulating quantum contextuality in the Peres-Mermim scenario and show that three internal states are necessary for a perfect simulation of  quantum behaviors. 
This shows that  reproducing the results of sequential measurements on a two-qubit system requires more memory than the information-carrying capacity of the system, given by the Holevo bound. In ref. \cite{KWB18} it is shown that  contextuality in a quantum sub-theory puts a lower bound on the cardinality of the  state space used in any classical  simulation of  this sub-theory. As a consequence of their result, the authors prove that   the minimum amount of bits necessary to simulate the $n$-qubit stabilizer sub-theory grows quadratically with $n$, in contrast with the qudit case with $d$ an odd prime \cite{DOBBR17}, where an efficient simulation that scales linearly in $n$ can be constructed from a particular choice of Wigner representation, which is always positive for stabilizer states.

\section{Conclusion}
\label{sec:conclusion}

In this contribution, we reviewed some of the recent developments towards a unified resource theory for contextuality. 
Although these results highlight contextuality as a possible operational resource, the understanding of the connection of these constructions with
practical applications is still in its infancy. This is the most important and the
most challenging ingredient in a resource theory. Other important question regarding contextuality as a resource remain open, such as the possibility of contextuality distillation, the role of catalysts, conversion rates and the possibility of finding an explicit parametrization of larger classes of free operations for contextuality.
It is also important to investigate the role of other forms of non-classicality in quantum advantage \cite{DA18,MK18} and, although quantifiers can be adapted to the contextuality-by-default framework of ref. \cite{DKC16}, it is still an open problem to find a version of the non-contextual wirings to this extended notion of contextuality.

We point out that although the main application of a resource theory is to understand the role of a physical property as an operational resource, this construction can be interesting on its own and it can  give  insight about the physical property under consideration. For example, in ref. \cite{DBAC18} the authors use  contextuality  quantifiers to explore the geometry of the set of behaviors, finding the approximate relative volume of the non-contextual set in relation to the non-disturbing set.

\subsection*{Acknowledgments}
The author thanks  Ad\'an Cabello, Ehtibar  Dzhafarov, Emily Tyhurst, Ernesto Galv\~ao, Jan-\AA{}ke Larsson, Jingfang Zhou,  Leandro Aolita, Marcelo Terra Cunha, Pawe\l{} Horodecki, Pawe\l{} Kurzy\'{n}ski,  Rui Soares Barbosa, Samsom Abramsky, Shane Mansfield and all the participants of the Winer Memorial Lectures at Purdue University for valuable discussions. The author thanks Ehtibar Dzhafarov, Maria Kon, and V\'ictor H. Cervantes for the organization of the event and Purdue University for its support and hospitality. The author  acknowledges financial support from the Brazilian ministries and agencies MEC and 
MCTIC,  INCT-IQ, FAPEMIG, and CNPq Universal grant  n. 431443/2018-1. 

\bibliographystyle{apalike}
\bibliography{biblio}

\end{document}